\documentclass[runningheads]{llncs}
\usepackage{graphicx}
\title{Vector Space Semantics for Lambek Calculus with Soft Subexponentials}%\thanks{Supported by organization x.}}
\titlerunning{Vector Space Semantics for $\SLLM$}
\author{Lachlan McPheat \and
	Hadi Wazni \and
	Mehrnoosh Sadrzadeh}
\authorrunning{McPheat et al.}
\institute{University College London, UK \\
\email{\{l.mcpheat, hadi.wazni.20, m.sadrzadeh\}@ucl.ac.uk}}

\usepackage{lscape}
\usepackage{proof}
\usepackage{stmaryrd}
\usepackage{amssymb}
\usepackage{mathtools}
\usepackage{float}
\usepackage{amstext}
\usepackage{cancel}
\usepackage{hyperref}
\usepackage{enumerate}
\usepackage{array}
\usepackage{mathrsfs}
\usepackage{tikz}
\usepackage{booktabs}
\usepackage{makecell}% http://ctan.org/pkg/makecell
\usepackage[all,cmtip]{xy}

\newcommand{\SLLM}{\mathbf{SLLM}}
\newcommand{\blstar}{\mathbf{!L^*}}
\newcommand{\semantics}[1]{[\![#1]\!]}
\newcommand{\csemantics}[1]{(\!|#1|\!)}
\newcommand{\ov}[1]{\overrightarrow{#1}}
\newcommand{\ol}[1]{\overline{#1}}
\newcommand{\id}{\mathrm{id}}
\newcommand{\ev}{\mathrm{ev}}

\newcommand{\C}{\mathcal{C}}
\newcommand{\fdVect}{\mathbf{fdVect}}

\newcommand{\bs}{\backslash}
\renewcommand{\epsilon}{\varepsilon}
\renewcommand{\phi}{\varphi}
\newcommand{\R}{\mathbb{R}}
\newcommand{\N}{\mathbb{N}}
\begin{document}

\maketitle
\begin{abstract}
We develop a vector space semantics for Lambek Calculus with Soft Subexponentials, apply the calculus to construct compositional vector interpretations for parasitic gap noun phrases and discourse units with anaphora and ellipsis, and experiment with the constructions in a distributional sentence similarity task. As opposed to previous work, which used Lambek Calculus with a Relevant Modality the calculus used in this paper uses a bounded version of the modality and is decidable.  The vector space semantics of this new modality allows us to meaningfully define contraction as projection and provide a linear theory behind what we could previously only  achieve via nonlinear maps.
\end{abstract}

\section{Introduction}

The origin of  logics that add copying and movement modalities to the Lambek calculus can be traced back to the work of Morrill et al \cite{barry1991}, who used restrictions on  linear logic's permutation and contraction rules in order to reason about larger fragments of natural language. Modalities were also added to Lambek calculus by Moortgaat \cite{moortgat1996multimodal} by means of indexed families of connectives that  would primarily control different forms of associativity. Later Moot explored the use of indexed modalities in linear logic and developed proof nets for them \cite{Moot2002}, a modality for controlled permutation was amongst this set of indexed modalities.  In  \cite{jager1998multi}, J\"{a}ger presented a multimodal Lambek calculus that modelled anaphora and ellipsis. Later in his book \cite{jager2006anaphora}, he put forward an alternative unimodal edition for  the same phenomena. The preferred semantics for these logics has always been a traditional Montague-style semantics, until a vector space interpretation was provided in  \cite{WijnSadr2019}. Here,  the authors develop a vector space interpretation for a  Lambek calculus with  restricted permutation and contraction rules and test the efficacy of a J\"{a}ger style treatment of verb-phrase ellipsis  on a disambiguation task of  distributional semantics \cite{Firth1957}. This  work led to a unified type-logical  distributional framework for co-reference resolution in natural language \cite{wijnholds-sadrzadeh-2019-evaluating}, but its underlying logic was undecidable and also could not distinguish between the strict and sloppy readings of  anaphora with ellipsis examples.

 Lambek calculus with a Relevant Modality  $\blstar$ \cite{Kanovich2016} is a recent  logic which follows suit from the work of  J\"{a}ger and Morrill.  %TODO add disclaimer
 In previous work, we developed sound  categorical and vector space semantics for this logic \cite{mcpheat2020categorical}. We demonstrated how it can model anaphora and ellipsis and can also distinguish between the sloppy vs strict readings \cite{mcpheat2021}. Our categorical semantics $\C(\blstar)$ is a monoidal biclosed category $(\C,\otimes, I, \Leftarrow,\Rightarrow)$, whose objects are the formulas of $\blstar$ and whose morphisms are its derivable sequents. We  require this category to be equipped with a coalgebra modality \cite{St2006} i.e. a  monoidal comonad $(!,\delta,\epsilon)$ on $\C$ and with diagonal natural transformation $\Delta :\ ! \to \ ! \ \otimes \ !$. This is the modality responsible for interpreting contraction. We further require that the coalgebra modality makes the objects of the category permute, i.e. we require a natural isomorphism $1_\C \otimes \ ! \ \cong \ ! \ \otimes 1_\C$. %TODO refer to ! from blute paper still? even though it's a mistake...
	The vector space semantics of $\blstar$, again from \cite{mcpheat2020categorical}, is a monoidal biclosed functor $\semantics{ \ }$ from $\C(\blstar)$ to the category of finite dimensional real vector spaces, $\fdVect$, maping formulas $A$ of $\blstar$ (objects of $\C(\blstar)$) to finite dimensional real vector spaces $\semantics{A}$, and derivable sequents $\Gamma \longrightarrow A$ of $\blstar$ (morphisms of $\C(\blstar)$) to linear maps $\semantics{\Gamma} \to \semantics{A}$. We instantiated the categorical model in $\fdVect$ by interpreting $!$ in several different ways, the most classical being a Fermionic Fock Space functor, originally used  in \cite{Blute1994} to interpret the $!$ modality of linear logic. 
	
	The calculus we depended on for this work, $\blstar$, however, was proven undecidable in the same paper it was introduced which has posed a challenge for all the work we have done on top of it. Thankfully, the creators of $\blstar$  have since defined a new logic with a similar expressive power but with a decidable fragment, referred to by $\SLLM$ \cite{Kanovich2020}. In $\SLLM$ permutation and contraction (now called `multiplexing') are separated into two different modalities, whose logical rules are inspired by Light and Soft linear logic \cite{Girard1998,Lafont2004}. It is this logic we soundly interpret in $\fdVect$ in this paper in the style of \cite{mcpheat2020categorical}.
	In doing so, we also managed to produce an interpretation for copying vectors, e.g. a vector $\ov{v}$ can now essentially be fully copied into two vectors $\ov{v} \otimes \ov{v}$ using a simple projection that was made possible due to the presence of accessible bounds. The map $\ov{v} \mapsto \ov{v} \otimes \ov{v}$ is non linear, and in previously we had to be content with linear approximations of it, e.g. maps that would only copy the bases.  Achieving full copying has been a surprising bonus when working in $\SLLM$.
	As with our previous work, we demonstrate the applicability of the logic to natural language by modelling anaphora and ellipsis and deriving  distinct interpretations of strict and sloppy examples. 
	
	No vector space semantics is complete without being accompanied by  large scale experimentation. So as in previous work, we also implement our model using neural word embeddings Word2vec, FastText, and BERT.  In previous work \cite{mcpheat2021} we experimented with a distributional disambiguation task and here we experiment with a distributional similarity task. We build type-driven vector representations for sentences with  verb phrase elliptical phrases of the similarity dataset of \cite{wijnholds-sadrzadeh-2019-evaluating} and compare the results of  the  linear copying map suggested in this paper with the different non-linear copying maps of previous work, with the non compositional verb-only, compositional grammar unaware  additive,  and BERT baselines. 
	
	%As in previous work,  one also gets a sound and complete -- by construction -- categorical semantics, generated over the syntax and sequent rules of the logic. Interestingly,  due to the new introduction rules for the modalities, this categorical semantics comes without a few undesirable properties, e.g. that the modalities are not \emph{lax or even op-lax monoidal}, i.e.  we do not have any primary compositional rules for them.  Lack of these,  makes finding complete models a harder challenge than in previous work and as a result the vector space model developed here  is incomplete (currently not an issue for linguistic applications).  Finding complete models does however merit further study from a logician's point of view. Investigating completeness possibly through centres of monoidal functors over finite dimensional vector spaces with a Hopf structure  is our work in progress. Experimenting with the setting on the Winograd Schema Challenge \cite{levesque2012} using  the plausibility models of \cite{polajnar2014b}  is another work in progress. 
	
\begin{table}[t!]
		\[
		A ::= A\in At \mid A \cdot A \mid A/A \mid A\bs A \mid \ !A \mid \nabla A
		\]
		\[\begin{array}{ll}
		\,\\\\
		\infer[I]
			{A\longrightarrow A}
			{}
		\\\\
		\infer[\bs_L]
			{\Sigma_1,\Gamma, A\bs B, \Sigma_2 \longrightarrow C}
			{\Gamma \longrightarrow A & \Sigma_1, B, \Sigma_2 \longrightarrow C}
		& \qquad
		\infer[\bs_R]
			{\Gamma \longrightarrow A\bs B}
			{A,\Gamma \longrightarrow B}
		\\\\
		\infer[/_L]
			{\Sigma_1, B/ A,\Gamma, \Sigma_2 \longrightarrow C}
			{\Gamma \longrightarrow A & \Sigma_1, B, \Sigma_2 \longrightarrow C}
		& \qquad
		\infer[/_R]
			{\Gamma \longrightarrow B/A}
			{\Gamma,A \longrightarrow B}
		\\\\
		\infer[\cdot_L]
			{\Gamma_1, A\cdot B, \Gamma_2 \longrightarrow C}
			{\Gamma_1, A, B, \Gamma_2 \longrightarrow C}
		& \qquad
		\infer[\cdot_R]
			{\Gamma_1,\Gamma_2 \longrightarrow A\cdot B}
			{\Gamma_1 \longrightarrow A 
			& 
			\Gamma_2 \longrightarrow B}
		\\\\
		\infer[!_L]
			{\Gamma_1,!A, \Gamma_2 \longrightarrow B}
			{\Gamma_1,\overbrace{A,A,\ldots, A}^{n\text{ times}}, \Gamma_2 \longrightarrow B} 
		& \qquad
		\infer[!_R]
			{!A \longrightarrow !B}
			{A \longrightarrow B}
		\\\\
		\infer[\nabla_L]
			{\Gamma_1,\nabla A, \Gamma_2 \longrightarrow B}
			{\Gamma_1,A, \Gamma_2 \longrightarrow B}
		& \qquad
		\infer[\nabla_R]
			{\nabla A\longrightarrow \nabla B}
			{A\longrightarrow B}
		\\\\
		\infer[perm]
			{\Gamma_1,\nabla A,\Gamma_2, \Gamma_3 \longrightarrow B}
			{\Gamma_1,\Gamma_2,\nabla A, \Gamma_3 \longrightarrow B}
		& \qquad
		\infer[perm']
			{\Gamma_1,\Gamma_2,\nabla A, \Gamma_3 \longrightarrow B}
			{\Gamma_1,\nabla A,\Gamma_2, \Gamma_3 \longrightarrow B}
		\end{array}\]
		\caption{Formulas and rules of $\SLLM$. Where $1\leq n\leq k_0$.}\label{tab:SLLMrules}
	\end{table}

\section{$\SLLM$, Lambek Calculus with Soft Subexponentials}\label{subsec:SLLM}
		 $\SLLM$ is a cut-free logic with a decidable fragment  and a directed proof system. We recall its  definition from \cite{Kanovich2020} in   table \ref{tab:SLLMrules}. The decidable fragment is found when one chooses a global bound ($k_0$) on the number of formulas you may contract using the multiplexing rule (denoted by $!_L$). %, the $\SLLM$-equivalent of contraction.
		Note that there are $k_0$ different instances of the multiplexing rule for a given bottom sequent $\Gamma_1,!A,\Gamma_2 \longrightarrow B$; one for each $1\leq n \leq k_0$. This determines the number of instances ($n$) of the formula $A$ that is being contracted (or `multiplexed') into $!A$. 
			Also notice the separation of multiplexing and permutation connectives and rules. A priori there is no reason to assume that the two properties of contractibility and permutation should coincide as they do in \cite{Kanovich2016}.

\section{Vector Space Semantics of $\SLLM$}\label{sec:VSS}
	In the following subsections we first inductively define a vector space semantics for $\SLLM$, prove its soundness, and then in detail explain how to construct vectors which effectively let us copy using projection maps. The interpretation of the subexponential $!$ is an adaptation of the tensor algebras of  \cite{mcpheat2020categorical}, which in turn originated from \cite{Blute1994}. The global multiplexing bound of  $\SLLM$ allows us to use a truncated form of the tensor algebras, we exploit this in the subsection following the proof of proposition \ref{prop:soundness}, and show how it helps us achieve  `full copying' as a linear map.
	
	A great deal of the definition of our semantics uses tensor products and dual spaces, which we recall in the following definitions.
	
	\begin{definition}\label{def:directSum}
		The \textbf{direct sum} (Cartesian Product, direct product) of two vector spaces $V,W$ is denoted $V\oplus W$ and contains vectors of the form $(v,w)$ where $v\in V$ and $w\in W$. It is a quick exercise to check that the interchange map $V\oplus W \to W \oplus V :: (v,w)\mapsto (w,v)$ is a linear isomorphism\footnote{That is, a linear map which is bijective.}, showing that $V\oplus W \cong W \oplus V$.
	\end{definition}
	
	\begin{definition}\label{def:tensorProduct}
		A \textbf{tensor product} of two vector spaces $V,W$ is a vector space $V \otimes W$ such that for any \textit{bilinear} map $h : V \oplus W \to U$ for some vector space $U$, there is a unique \textit{linear} map $\tilde{h}: V\otimes W \to U$ such that the original map $h$ is equal to the composite of $\tilde{h}$ and the canonical map $V\oplus W \to V\otimes W :: (v,w) \mapsto v\otimes w$. 
	\end{definition}

	Note also that since we have $V\oplus W \cong W \oplus V$, one can derive that $V\otimes W \cong W \otimes V$ from definition \ref{def:tensorProduct}.
	
	\begin{definition}\label{def:dualSpace}
		Given a vector space $V$, its \textbf{dual vector space} is the set of linear functionals on $V$, that is the set $\{f : V \to \mathbb{R} \mid f \ \text{linear}\}$, which we denote by $V^*$. We leave it to the reader to verify that $V^*$ is indeed a vector space.
	\end{definition}
	
	The last thing we recall before defining the vector space semantics is that the set of linear maps from a space $V$ to a space $W$ is itself a vector space, and is isomorphic to the space $V^* \otimes W$. With definitions \ref{def:tensorProduct} and \ref{def:dualSpace} and the note about spaces of maps under our belt, we may proceed to define vector space models of $\SLLM$.
	
%	The above definition is written in a category theoretic style, and is ultimately more useful for the definition of our semantics, but for the computationally oriented, there is another way to think about tensor products which may be more useful. Given bases $(v_i)_i$ of $V$ and $(w_j)_j$ of $W$, one may think of $V\otimes W$ as the vector space spanned by the symbols $(v_i \otimes w_j)_{ij}$, with a slightly peculiar arithmetic and scaling. Namely, for any scalar $\lambda \in \R$ and arbitrary vector $v\otimes w \in V\otimes W$. we have $\lambda (v\otimes w) = \lambda v \otimes w = v \otimes \lambda w$. We also have, for any vectors $v,v' \in V$ and $w,w'\in W$ that $(v+v')\otimes w = v\otimes w + v'\otimes w$ and $v \otimes (w + w') = v\otimes w + v \otimes w'$. The space $V\otimes W$ clearly has dimension equal to $\dim V \dim W$, since there are $\dim V \dim W$ many basis elements.

	\begin{definition}\label{def:VSSDef}
	Vector space models of $\SLLM$ consist of a pair of maps, one mapping formulas of $\SLLM$ to finite dimensional real vector spaces, the other mapping proofs of $\SLLM$ to linear maps.
	We define a vector space model, denoted $ \semantics{\ } : \SLLM \to \fdVect$, and define it inductively on formulas and derivable sequents below.
	\[\begin{array}{c}
		\semantics{A} := V_A \in \fdVect_\R, \text{ for atomic formulas } A, \\ \\
		\semantics{A, B} = \semantics{A \cdot B} := \semantics{A} \otimes
	\semantics{B} \qquad \semantics{!A} := T_{k_0}\semantics{A} \qquad \semantics{\nabla A} :=\semantics{A}
	\\ \\
		\semantics{A \bs B} := \semantics{A}^* \otimes \semantics{B} \qquad \qquad 
		\semantics{B / A} := \semantics{A}^* \otimes \semantics{B} 
	\end{array}\]
	where $k_0 \in \N$ is the multiplexing bound and   $T_{k_0}\semantics{A}$ is the $k_0$-th truncated Fock space of $\semantics{A}$, defined on a vector space $V$ as:
	\[
	T_{k_0}V := \bigoplus_{i=0}^{k_0}V^{\otimes i} = \R \oplus V \oplus (V \otimes V) \oplus \cdots \oplus V^{\otimes k_0}.
	\]
	Proofs of sequents $\Gamma \longrightarrow A $ in $\SLLM$ are interpreted as linear maps \hbox{$\semantics{\Gamma} \longrightarrow \semantics{A}$}. 
	\end{definition}
	%The details are provided in the following section.
\begin{definition}\label{def:truth}
A rule of  \ $\SLLM$ is sound in  a vector space model iff whenever the interpretation of the top part is a linear map, so is the interpretation of the bottom part. $\SLLM$ is sound in a vector space model iff all its rules are.
\end{definition}

		\begin{proposition} \label{prop:soundness} $\SLLM$ is sound in $\semantics{\ }$ of definition \ref{def:VSSDef}.
		\end{proposition}
		
		\begin{proof} We proceed by a case analysis on the soundness of the rules. \\
		\\
		\textbf{{Axiom}:} \
			The axiom of $\SLLM$ is interpreted as the existence of identity maps and is immediately sound.
		\\
		\\
		$\mathbf{/_L, \bs_L} $ are interpreted as:
		\[
		\infer[]
		{g^f:\semantics{\Delta_1}\otimes \semantics{B}\otimes \semantics{A}^* \otimes \semantics{\Gamma} \otimes \semantics{\Delta_2} \longrightarrow \semantics{C}}
		{f: \semantics{\Gamma}\longrightarrow \semantics{A} & 
			g: \semantics{\Delta_1}\otimes \semantics{B} \otimes \semantics{\Delta_2} \longrightarrow \semantics{C}}
		\] 
		and 
		\[
		\infer[]
		{^fg:\semantics{\Delta_1}\otimes \semantics{\Gamma} \otimes \semantics{A}^* \otimes \semantics{B}\otimes  \semantics{\Delta_2} \longrightarrow \semantics{C}}
		{f: \semantics{\Gamma}\longrightarrow \semantics{A} & 
			g: \semantics{\Delta_1}\otimes \semantics{B} \otimes \semantics{\Delta_2} \longrightarrow \semantics{C}}
		,\]
		where $g^f$ and $^fg$ are the defined below. All of the maps in the definition are linear, thus making $g^f$ and $^fg$ linear too, as required.\footnote{The map $\ev$ refers to the evaluation map. For example, $\ev_{V \Rightarrow W}(f \otimes w)$ is defined to be equal to $f(w)$ and $\ev_{W\Leftarrow V}(w \otimes g)$ is defined to be equal to $g(w)$ given $f\in V\otimes W^*, g \in W^*\otimes V, w\in W$.}
		\[\begin{array}{c}
		g^f := g \circ (\id_{\semantics{\Delta_1}} \otimes \ev_{\semantics{B}\Leftarrow\semantics{A}} \otimes \id_{\semantics{\Delta_2}}) \circ (\id_{\semantics{\Delta_1}} \otimes \id_{\semantics{B}} \otimes \id_{\semantics{A}^*} \otimes f \otimes \id_{\semantics{\Delta_2}})
		\\ \\
		^fg := g \circ (\id_{\semantics{\Delta_1}} \otimes \ev_{\semantics{B}\Rightarrow\semantics{A}} \otimes \id_{\semantics{\Delta_2}}) \circ (\id_{\semantics{\Delta_1}} \otimes f \otimes \id_{\semantics{A}^*} \otimes \id_{\semantics{B}} \otimes \id_{\semantics{\Delta_2}})
		\end{array}\]
		\\
	 $\mathbf{/_R}, \mathbf{\bs_R} $ are interpreted using the \emph{tensor-hom adjunction}, also referred to as the \emph{left and right currying}. The semantics of $/_R$ is as follows:
		\[\infer[]
		{\Lambda^r(f) : \semantics{\Gamma} \longrightarrow \semantics{B} \otimes \semantics{A}^*}
		{f: \semantics{\Gamma}\otimes\semantics{A} \longrightarrow \semantics{B}}
		\]
		where for vectors $\gamma \in \semantics{\Gamma}$, we have a linear map $\Lambda^r(f)(\gamma)  \in \semantics{B}\otimes\semantics{A}^*$ given by $\Lambda^r(f)(\gamma)(a) := f(\gamma \otimes a)$. 
		Similarly for the $(\bs R)$ rule we have:
		\[\infer[]
		{\Lambda^l(f): \semantics{\Gamma} \longrightarrow \semantics{A}^*\otimes \semantics{B}}
		{f: \semantics{A}\otimes\semantics{\Gamma} \longrightarrow \semantics{B}}
		\]
		where $\Lambda^l(f)$ is defined analogously.
		\\
		\\
		$\mathbf{\cdot_L}, \mathbf{\cdot_R}$  are immediately sound by the identification of the comma and the $\cdot$ in the semantics. Explicitly:
		\[
		\infer[]
		{f: \semantics{\Gamma_1}\otimes\semantics{A}\otimes\semantics{B}\otimes \semantics{\Gamma_2}\to \semantics{C}}
		{f: \semantics{\Gamma_1}\otimes\semantics{A}\otimes\semantics{B} \otimes \semantics{\Gamma_2}\to \semantics{C}}
		\qquad
			\infer[]
			{f\otimes g: \semantics{\Gamma_1}\otimes\semantics{\Gamma_2} \to \semantics{A}\otimes\semantics{B}}
			{f: \semantics{\Gamma_1} \to \semantics{A} & g: \semantics{\Gamma_2}\to \semantics{B}}.
		\]
		\\ 
		$\mathbf{!_L} $ is interpreted as:
			\[\infer[]
			{C_n(f): \semantics{\Delta_1}\otimes T_{k_0}\semantics{A} \otimes \semantics{\Delta_2} \longrightarrow \semantics{B}}
			{f: \semantics{\Delta_1}\otimes \overbrace{\semantics{A}\otimes \semantics{A} \cdots \otimes \semantics{A}}^{n-\text{times}} \otimes \semantics{\Delta_2} \longrightarrow \semantics{B}}
			\]
			where $C_n(f)$ is a linear map defined using the $n$-th projection $\pi_n: T_{k_0}(\semantics{A}) \to \semantics{A}^{\otimes n}$ and as $f \circ (\id_{\semantics{\Delta_1}} \otimes \pi_n \otimes \id_{\semantics{\Delta_2}}).$
		\\ 
		\\
	 $\mathbf{!_R}$ is interpreted as an application of $T_{k_0}$. That is, 
		 \[
		 	\infer[]
		 		{T_{k_0} f:  T_{k_0}\semantics{A} \to T_{k_0}\semantics{B}}
		 		{f : \semantics{A} \to \semantics{B}}
		 \]
		 where $T_{k_0}(f)$ is a linear map, defined as:
		 \[T_{k_0}(f)\left(\sum_{i = 0}^{k_0} \bigotimes_{j=1}^{n_i}v_j\right) := \sum_{i = 0}^{k_0} \bigotimes_{j=1}^{n_i}f(v_j).\]
		 For an easier example of how $T_{k_0}(f)$ is applied, consider a vector $u + v \otimes w \in T_{k_0}\semantics{A} $ which is mapped to $f(u) + f(v)\otimes f(w) \in T_{k_0}\semantics{B}$. %TODO have we consistently used sum notation for vectors in a biproduct space??
		\\
		\\
	$\mathbf{\nabla_R}, \mathbf{\nabla_L}$ are interpreted trivially, since $\nabla$ is interpreted as identity on formulas. One sees this explicitly when writing out the interpretations of $\nabla_L,\nabla_R$ in the semantics. Consider $\nabla_L$, which  corresponds to 
		\[
			\infer[]
				{f': \semantics{\Gamma_1}\otimes \semantics{A}\otimes \semantics{\Gamma_2} \to \semantics{B}}
				{f: \semantics{\Gamma_1}\otimes \semantics{A}\otimes \semantics{\Gamma_2} \to \semantics{B}}.
		\]
		Since we interpret $\nabla A $ as $\semantics{\nabla A} = \semantics{A}$ there is no distinction remaining between the top and bottom rules, meaning that the a priori different interpretation $f'$ is in fact equal to $f$. The case for $\nabla_L$ is slightly simpler, and we encourage the reader to write it out for themselves.
		\\
		\\ 
		$\mathbf{perm}, \mathbf{perm'}$ are interpreted using the symmetry of $\otimes$ in $\fdVect$.
		\qed
	\end{proof}
	
	\paragraph{\bf Interpreting Projection as Copying.}%\label{subsec:projAsCopy}	
	Syntactically speaking,  the multiplexing rule does the job of  contraction for $\SLLM$. This new syntax, however,  allows its vector interpretation  to greatly differ  from the vector space interpretation of  contraction,  developed for $\blstar$  in \cite{mcpheat2020categorical}. This new interpretation  provides us with `full' copies. %, when using a specific input format on vectors in $\semantics{!A}$-type spaces. 
	In the framework of \cite{mcpheat2020categorical} it was only possible to achieve an approximation of (the nonlinear) full copying $v\mapsto v\otimes v$, a problem that is remedied here by exploiting the global multiplexing bound which in turn induced a bound on our interpretation of the $!$-modality.
	
	 Given any $n$ in the range $1\leq n \leq k_0$,  multiplexing  $T_{k_0}V \to V^{\otimes n}$ is interpreted as the projection map $\pi_n$. In general, this map will not appear to copy any of its inputs and contrarily it seems  to `forget' most of them.  This `forgetting' is exactly what we exploit when editing the shape of the input vectors using the simple $\sim$-construction: given any vector $v \in V$, we define $\tilde{v}\in T_{k_0}V$ as 
		\[\tilde{v} := \sum_{i=0}^{k_0}v^{\otimes i} = 1 + v + v\otimes v + \cdots + \overbrace{v\otimes \cdots \otimes v}^{k_0 \text{times}}
		\]
		%(1,v,v\otimes v, v\otimes v \otimes v, \ldots , v^{\otimes k_0}).\]
	By applying the $n$-th projection map to $\tilde{v}$ we get $\pi_n(\tilde{v}) = v^{\otimes n}$. 
	To achieve `full' copying, take $n=2$ and we obtain $\pi_2(\tilde{v}) = v \otimes v$. So we arrive at an interpretation of contraction which at the surface level looks like copying of $\sim$-ed vectors but which is in fact a simple, even canonical linear map.   
	Clearly we are not making the map $v \mapsto v\otimes v$ linear, nor are we truly copying our input (i.e. we are not achieving $\tilde{v}\mapsto \tilde{v}\otimes \tilde{v}$), instead, we think of the map $\pi_2:\tilde{v} \mapsto v \otimes v$ as copying $v$ by going through the layers of its bounded Fock space.

		\section{Categorical Semantics}\label{sec:catSem}
	
	For readers not familiar with category theory, this section may be skipped entirely. We present the categorical semantics as it is a useful theoretical object for defining further kinds of semantics, and broadens the audience for this work to include applied category theorists. We provide a brief introduction to the category theory used in this paper, but also recommend reading a brief introduction to category theory such as \cite{leinster2016}, as well as descriptions of categories in linguistic analysis, such as \cite{Lambek1988}. We will try to introduce the necessary categorical machinery next, before defining our categorical semantics.
	
	\begin{definition}
	
	By \textbf{monoidal biclosed category} we mean a set of \textbf{objects} and a set of \textbf{morphisms}.
	To define the objects, one first fixes a set of atomic symbols $\{A, B, C, \ldots\}$ and add a distinguished symbol $I$, and three binary operations $\otimes, \Rightarrow, \Leftarrow$. The full set of objects is then defined to be the free $(\otimes, \Rightarrow, \Leftarrow)$-algebra over the given set of atomic symbols. Thus our objects may look like $I$, $A$, $A\otimes A$, $A\Rightarrow B$ $B\Leftarrow (A\otimes B)$ and so on.
	
	Once we are familiar with the objects, we may define morphisms as a set of arrows pointing from the \textbf{domain} to the \textbf{codomain} (usually denoted $A\to B$ or $B\to C \otimes D$ etc). We may label morphisms by writing, for example, $f:A\to B$. For every object $A$, there is a unique identity morphism $\id_A :A \to A$. We may compose pairs of arrows if the codomain of one agrees with the domain of the other. That is: 
		\[(f: A\to B, \quad g: B \to C) \longmapsto g\circ f : A \to C.\]
	Composing with identity morphisms does nothing. That is, for any $f:A \to B$ we have 
		\[\id_B \circ f = f = f \circ \id_A.\]
	We also define monoidal products of morphisms as:
		\[(f: A\to B, \quad g: B \to C) \longmapsto g\otimes f : B\otimes A \to C \otimes B.\]
	Both $\circ$ and $\otimes$ are associative, and interact according to the following rule 
		\[(f\otimes g) \circ (f'\otimes g') = (f\circ f') \otimes (g\circ g'),\]
	given that the domains and codomains agree in the necessary way. When it is clear, one often omits the composition symbol, thus $g\circ f$ would be written as $gf$.
	\end{definition}
	
	Some typical examples of monoidal biclosed categories are the category of sets and functions, $\mathbf{Set}$, where the objects are sets, the monoidal product is the cartesian product, and for any two sets $A, B$ we define $A\Rightarrow B := \{f:A \to B\}$, i.e. $A\Rightarrow B$ is the set of functions from $A$ to $B$. In $\mathbf{Set}$ it so happens that $A\Rightarrow B \cong B \Leftarrow A$, although this is not necessarily the case in an arbitrary monoidal biclosed category. 
	An example of a monoidal biclosed category which we make use of is the category of finite dimensional (real) vector spaces $\fdVect$, where the objects are finite dimensional vector spaces, and the morphisms are linear maps. In $\fdVect$ the monoidal product is the tensor product, and the object $V\Rightarrow W$ is the set of linear maps from $V$ to $W$.
	
	We also recall briefly what a functor is. 
	
	\begin{definition}\label{def:functor}
	Given two (monoidal biclosed) categories $\C, \mathcal{D}$ a \textbf{functor} $F: \C \to \mathcal{D}$ is a pair of functions, one mapping objects of $\C$ to objects of $\mathcal{D}$ the other mapping the morphisms of $\C$ to morphisms of $\mathcal{D}$, such that given a morphism $f : A \to B$ in $\C$, we have $F(f) : F(A) \to F(B)$ in $\mathcal{D}$.
	We also require composition to be preserved, i.e. $F(g\circ f) = F(g) \circ F(f)$.
		
	\end{definition}
	A typical example of a functor is the identity functor $1_\C : \C \to \C$, which is defined on any category $\C$ as $1_\C(A):= A$ on objects and $1_\C(f) := f$ on morphisms.
	
	\begin{definition}\label{def:naturalTransformation}
		Given two functors $F,G : \C \to \mathcal{D}$, A \textbf{natural transformation} $\alpha$ from $F$ to $G$, denoted $\alpha : F \to G$, is a collection of morphisms indexed by objects of $\C$, $(\alpha_A : F(A) \to G(A))_{A\in \C}$ such that for any morphism $f:A\to B$ in $\C$, the following diagram commutes:
		\[
			\xymatrix{F(A) \ar[r]^{F(f)} \ar[d]^{\alpha_A} & F(B) \ar[d]^{\alpha_B} \\
			G(A) \ar[r]^{G(f)} & G(B)
			}
		\]
		That is, $G(f) \circ \alpha_A = \alpha_B \circ F(f)$. 
	\end{definition}
	
	Having recalled what monoidal biclosed categories and functors and natural transformations are, we may define a categorical semantics for $\SLLM$ with relative ease.
	
	\begin{definition}\label{def:categoricalSemantics}
	Our \textbf{categorical semantics of the decidable fragment of $\SLLM$}, with global multiplexing bound $k_0$ is a monoidal biclosed category $\mathcal{C}(\SLLM)$, equipped with two functors $M,P:\mathcal{C}(\SLLM) \to \mathcal{C}(\SLLM)$, defined below.
	
	The objects of $\C(\SLLM)$ are the formulas of $\SLLM$, and the morphisms $A\to B$ of $\C(\SLLM)$ are derivations of the sequent $A\longrightarrow B$ in $\SLLM$.
%	\begin{color}{red} %TODO add "proofs are derivations of SLLM"
%		The equivalence on proofs in $\SLLM$ is based on that found in \cite{mellies2009}, namely exchange invariance, and associativity of cut. Obviously our exchange invariance only applies to $\nabla$-formulas...
%	\end{color}
	The two functors are given names $M$ (for \textbf{multiplexing}) and $P$ (for \textbf{permutation}), which in turn have the following three structures.
	\begin{enumerate}
		\item For each $n=1,\ldots k_0$ and every object $A\in\C(\SLLM)$ there is a map $\delta_n : MA \to A^{\otimes n}$.
		\item There is a counit for $P$, that is, a natural transformation $\epsilon : P \to 1$.
		\item There is a natural isomorphism\footnote{A natural isomorphism is a natural transformation whose components are all isomorphisms.} $\sigma: 1\otimes P \cong P\otimes 1$.
	\end{enumerate}
	
	Note that the global multiplexing bound of $\SLLM$, $k_0$, determines the number of morphisms $\delta_n : MA \to A^{\otimes n}$ in our multiplexing functor.
	
	%we define an interpretation map $\csemantics{\ }:\SLLM \to \C(\SLLM) $ inductively on formulas and sequents and then show it soundly interprets the rules of $\SLLM$.
	For every atomic formula $A$ of $\SLLM$ we prescribe an object $\csemantics{A} := C_A$. For the complex formulas of $\SLLM$ we interpret them as follows:
	\[\begin{array}{c}
		\csemantics{A,B}= \csemantics{A\cdot B} := \csemantics{A}\otimes \csemantics{B} 
		\qquad \csemantics{!A} := M\csemantics{A} 
		\qquad \csemantics{\nabla A} := P\csemantics{A}
		\\ \\
		\csemantics{A\bs B} := \csemantics{A}\Rightarrow \csemantics{B} 
		\qquad \csemantics{B/A} := \csemantics{B}\Leftarrow\csemantics{A}
	\end{array}\]
	
	\end{definition}
	
	We can now prove a proposition similar to our proposition \ref{prop:soundness}, that the  rules of $\SLLM$ are sound in $\C(\SLLM)$.  We sketch the proof for the more complicated rules and leave the rest for the attentive reader to fill in. 
	\begin{proof}[Sketch]
	$\mathbf{axiom}$ is interpreted as the identity morphism for objects of $\C(\SLLM)$.
	$\mathbf{\bs_L}$ is interpreted as:
		\[
			\infer[]
			{^fg: \csemantics{\Delta_1}\otimes \csemantics{\Gamma}\otimes \csemantics{A}\Rightarrow \csemantics{B}\otimes \csemantics{\Delta_2} \to \csemantics{C}}
			{f: \csemantics{\Gamma}\to \csemantics{A}
			&
			g: \csemantics{\Delta_1}\otimes\csemantics{B}\otimes \csemantics{\Delta_2} \to \csemantics{C}}
		\]
	where $^fg := g \circ (\id_{\csemantics{\Delta_1}}\otimes \ev_{\csemantics{A}\Rightarrow\csemantics{B}} \otimes \id_{\csemantics{\Delta_2}}) \circ (\id_{\csemantics{\Delta_1}}\otimes f \otimes \id_{\csemantics{\Delta_2}})$, much like the definition in the vector space semantics. %TODO check if the \ev notation is the same here as in VSS section
	$\mathbf{\bs_R} $ is interpreted immediately, via the same tensor-hom adjunction used for vector space semantics. In the category theory, this  comes naturally from the requirement that $\C(\SLLM)$ be monoidal closed. 
	%
%	Explicitly this is:
%	\[
%		\infer[]
%		{\bar{f}: \csemantics{\Gamma}\to \csemantics{A}\Rightarrow \csemantics{B}}
%		{f: \csemantics{A}\otimes\csemantics{\Gamma}\to \csemantics{B}}
%	\]
%	where $\bar{f}$ is the transpose of $f$ under the tensor-hom adjunction.
%	\\
%	\\
%	$\mathbf{\cdot_L}, \mathbf{\cdot_R}$ are interpreted exactly as in section \ref{prop:soundness}, since we again identify the interpretations of $\cdot$ and $,$ in the categorical semantics.
	$\mathbf{!_L}$ is interpreted in $k_0$ different instances, one for each $n=1,\ldots, k_0$. Explicitly, given $n$, this is:
	\[
		\infer[]
		{C_n(f): f:\csemantics{\Gamma_1}\otimes M\csemantics{A}\otimes \csemantics{\Gamma_2} \to \csemantics{B}}
		{f:\csemantics{\Gamma_1}\otimes\csemantics{A}^{\otimes n}\otimes \csemantics{\Gamma_2} \to \csemantics{B}}
	\]
	where $C_n(f)$ is defined to be $f\circ (\id_{\csemantics{\Gamma_1}}\otimes \delta_n \otimes \id_{\csemantics{\Gamma_2}})$.
	$\mathbf{!_R} :$ is interpreted immediately as the application of the functor $M$.
	$\mathbf{\nabla_L}$ is interpreted using the counit $\epsilon : P \to 1$ as follows:
	\[
		\infer[]
		{Q(f):\csemantics{\Gamma_1}\otimes P\csemantics{A}\otimes \csemantics{\Gamma_2} \to \csemantics{B}}
		{f: \csemantics{\Gamma_1}\otimes \csemantics{A}\otimes \csemantics{\Gamma_2} \to \csemantics{B}}
	\]
	where $Q(f)$ is defined to be $f\circ (\id_{\csemantics{\Gamma_1}}\otimes \epsilon \otimes \id_{\csemantics{\Gamma_2}})$. $\mathbf{\nabla_R}:$ is interpreted just like $!_R$ above, but as an application of $P$.
	$\mathbf{perm}, \mathbf{perm'}$ are interpreted using the natural isomorphism $\sigma :\id \otimes P \cong P\otimes \id$.
%	 as follows:
%	\[
%	\infer[]
%		{\sigma(f): \csemantics{\Gamma_1}\otimes P\csemantics{A}\otimes \csemantics{\Gamma_2}\otimes \csemantics{\Gamma_3}\to \csemantics{B}}
%		{f: \csemantics{\Gamma_1}\otimes\csemantics{\Gamma_2}\otimes P\csemantics{A} \otimes \csemantics{\Gamma_3}\to \csemantics{B}}
%	\]
%	where $\sigma(f)$ is defined as $f\circ (1_{\csemantics{\Gamma_1}}\otimes \sigma_{\csemantics{\Gamma_2},\csemantics{A}} \otimes 1_{\csemantics{\Gamma_3}})$
	\qed\end{proof}
	
	Note that since we have fixed a bound $k_0$ in the syntax of the modality $!$ see item 1 of definition \ref{def:categoricalSemantics}, the  categorical semantics defined above is not only sound, but also a complete model of $\SLLM$. This is an automatic construction akin to that of syntactic categories \cite{Johnstone} or free categories over a monoidal signature as in \cite{Selinger2010} and it lets you define semantics of $\SLLM$ in other categories $\mathcal{D}$ as structure preserving functors $\C(\SLLM)\to \mathcal{D}$, rather than tediously specifying a semantics in $\mathcal{D}$ directly from $\SLLM$. In fact, the vector space semantics $\semantics{\ }:\SLLM \to \fdVect$ of section \ref{sec:VSS} factors through $\csemantics{\ } : \SLLM \to \C(\SLLM)$.
	
	\section{$!$ and $\nabla$ in Natural Language Syntax and Semantics} \label{sec:examples}
	We demonstrate how to analyse the same  phenomena as in previous work \cite{mcpheat2020categorical} and \cite{mcpheat2021}, namely parasitic gaps, anaphora, (verb phrase) ellipsis, and a combination of the latter two. Again, $\SLLM$ distinguishes between strict and sloppy readings in the ambiguous anaphora with ellipsis cases and this distinction is also carried into the semantics. 
	
	\subsection{Parasitic Gaps}
		 The noun phrase \textit{Papers that John signed without reading}, with $\SLLM$ types\footnote{The choice of $n/n$ for the type of "reading" is from original work of \cite{Kanovich2016}.}:
		\[
		\begin{array}{c}
		\{
		(\text{Papers}: n),
		(\text{that}: (n\bs n)/(s/!\nabla n)),
		(\text{John}: n),
		(\text{signed}: n \bs s /n),\\
		(\text{without},((n\bs s)\bs(n\bs s))/n),
		(\text{reading}, n/n)
		\}
		\end{array}
		\]
		acquired the following $\SLLM$ derivation:
		\[
		\scalebox{0.85}{\infer[\bs_L]
			{n,(n\bs n)/(s/!\nabla n),n,n \bs s /n, ((n\bs s)\bs(n\bs s))/n, n/n \longrightarrow n}
			{\infer[\bs_L]{n,n\bs n \longrightarrow n}
				{\infer[]{n \longrightarrow n}{}
				& 
				\infer[]{n \longrightarrow n}{}}
				&
				\infer[/_R]{n,n \bs s /n, ((n\bs s)\bs(n\bs s))/n, n/n \longrightarrow s/!\nabla n}
				{\infer[!_L]{n,n \bs s /n,((n\bs s)\bs(n\bs s))/n, n/n, !\nabla n \longrightarrow s}
					{\infer[perm']{n,n \bs s /n,((n\bs s)\bs(n\bs s))/n, n/n, \nabla n,\nabla n \longrightarrow s}
					{\infer[perm']{n,n \bs s /n,((n\bs s)\bs(n\bs s))/n,\nabla n, n/n, \nabla n \longrightarrow s}
						{\infer[\nabla_L]{n,n \bs s /n,\nabla n,((n\bs s)\bs(n\bs s))/n, n/n, \nabla n \longrightarrow s}
							{\infer[\nabla_L]{n,n \bs s /n, n,((n\bs s)\bs(n\bs s))/n, n/n, \nabla n \longrightarrow s}
								{\infer[/_L]{n,n \bs s /n, n,((n\bs s)\bs(n\bs s))/n, n/n, n \longrightarrow s}
									{\infer[]{n\longrightarrow n}{}
									&
									\infer[/_L]{n,n \bs s /n, n,((n\bs s)\bs(n\bs s))/n, n \longrightarrow s}
										{\infer[]{n\longrightarrow n}{}
										&
										\infer[/_L]{n,n \bs s/n, n,(n\bs s)\bs(n\bs s) \longrightarrow s}
											{\infer[]{n\longrightarrow n}{}
											&
											\infer[\bs_L]{n,n\bs s,(n\bs s)\bs(n\bs s)\longrightarrow s}
												{\infer[]{n\bs s\longrightarrow n\bs s}{}
												&
												\infer[\bs_L]{n, n\bs s \longrightarrow s}
													{\infer[]{n\longrightarrow n}{}
													&
													\infer[]{s\longrightarrow s}{}}}}}}}}}}}}
			}
		}\]
	The linear map representing this derivation is:
	\small\[
	N \otimes (N^*\otimes N)\otimes(S\otimes (T_{k_0}N)^*)^*\otimes N \otimes N^*\otimes S\otimes N^* \otimes (N^*\otimes S)^*\otimes (N^*\otimes S)\otimes N^* \otimes N\otimes N^* \to N
	\]
	\normalsize
	\[
	\begin{array}{l}
	:: \ov{Papers}\otimes \ol{that} \otimes \ov{John} \otimes \ol{signed}\otimes\ol{without}\otimes\ol{reading}\\
		\mapsto 
		\ol{that}(\ov{Papers},\ov{without}(\ov{John}, \ol{signed}(-,-), \ol{reading}(-)))
	\end{array}\]
	simplified to the following by interpreting \textit{that} as evaluation:
	\[\ov{without}(\ov{John}, \ol{signed}(-,\ov{Papers}), \ol{reading}(\ov{Papers}))).\]
	
	\subsection{Pronominal Anaphora}
	Following our motivating example in \cite{mcpheat2021}, we consider the utterance \emph{John sleeps. He snores}, with the following $\SLLM$ types:
	\[\{(\text{\emph{John}}:\ !(\nabla n)),(\text{\emph{sleeps}}: n\bs s), (\text{\emph{He}}:\ \nabla n\bs n), (\text{\emph{snores}}:n\bs s)\}.\]
	
	Note the type of \textit{John} in this example is $!(\nabla n)$, whereas earlier it was simply $n$. In fact, one can always retrieve the type $n$ from $!(\nabla n)$ using the following type raising operation: 
	\[\infer[!_L]{\Gamma_1,!(\nabla n),\Gamma_2 \longrightarrow X}
		{\infer[\nabla_L]{\Gamma_1, \nabla n, \Gamma_2 \longrightarrow X}
			{\Gamma_1, n, \Gamma_2 \longrightarrow X}}
		\]

 The derivation corresponding to the syntax of this utterance in $\SLLM$ is:
	\[\scalebox{0.85}{
	\infer[contr_2]
	{!(\nabla n), n\bs s,\ \nabla n\bs n, n\bs s \longrightarrow s \cdot s}
	{\infer[perm]{\nabla n,\,\nabla n, n\bs s,\ \nabla n\bs n, n\bs s \longrightarrow s \cdot s}
		{\infer[\nabla_L]{\nabla n, n\bs s,\,\nabla n, \ \nabla n\bs n, n\bs s \longrightarrow s \cdot s}
			{\infer[\bs_L]{n, n\bs s,\,\nabla n, \ \nabla n\bs n, n\bs s \longrightarrow s \cdot s}
				{\infer[]{n \longrightarrow n}{}
				& 
				\infer[\bs_L]{s,\,\nabla n, \ \nabla n\bs n, n\bs s \longrightarrow s \cdot s}
					{\infer[\nabla_R]{\nabla n \longrightarrow \ \nabla n}
						{\infer{n \longrightarrow  n}{}}
					&
					\infer[\bs_L]{s,n,n\bs s \longrightarrow s \cdot s}
						{\infer[]{n \longrightarrow n }{}
						&
						\infer[\cdot_R]{s,s \longrightarrow s \cdot s}
							{\infer{s \longrightarrow  s}{}
							&
							\infer{s \longrightarrow  s}{}}}}}}}}
	}\]
	This derivation translates into the following linear map:
	\[\begin{array}{c}
			f: T_{k_0}N \otimes N^* \otimes S \otimes N^* \otimes N \otimes N^* \otimes S \longrightarrow S \otimes S
			\\
			:: \tilde{n}\otimes \nu \otimes s \otimes  \nu' \otimes n' \otimes \nu'' \otimes s' \longmapsto \nu (n) s \otimes \nu'(n)\nu''(n')s
	\end{array}\]
	
	Taking  $\widetilde{\text{\emph{John}}}\in T_{k_0}N$, $\overrightarrow{\text{\emph{sleeps}}} \in N^* \otimes S$, $\overrightarrow{\text{\emph{He}}}  \in N^* \otimes N$ and $\overrightarrow{\text{\emph{snores}}} \in N^*\otimes S$ we obtain the simplified version below:
	\[f\left(\widetilde{\text{\emph{John}}} \otimes\overrightarrow{\text{\emph{sleeps}}} \otimes \overrightarrow{\text{\emph{He}}} \otimes\overrightarrow{\text{\emph{snores}}} \right) = \overrightarrow{\text{\emph{sleeps}}}(\overrightarrow{\text{\emph{John}}}) \otimes \overrightarrow{\text{\emph{snores}}}(\overrightarrow{\text{\emph{He}}}(\overrightarrow{\text{\emph{John}}})).\]

	\subsection{Ellipsis}
	Consider the utterance \emph{John plays guitar. Lisa does too.} with $\SLLM$ types:
	\[\{
	(\text{\emph{John}}: n),(\text{\emph{plays guitar}}: !(\nabla (n\bs s))),(\text{\emph{Lisa}}: n), (\text{\emph{does-too}}:\ \nabla (n\bs s)\bs n\bs s)
	\}\]
	
	The derivation of this utterance  in $\SLLM$ is:
	\[\scalebox{0.85}{
	\infer[contr_2]
	{n,!(\nabla (n\bs s)),n,\ \nabla (n\bs s)\bs n\bs s \longrightarrow s \cdot s}
	{\infer[perm]{n,\ \nabla (n\bs s),\ \nabla (n\bs s),n,\ \nabla (n\bs s)\bs n\bs s \longrightarrow s \cdot s}
		{\infer[\nabla_L]{n,\ \nabla (n\bs s),n,\ \nabla (n\bs s), \ \nabla (n\bs s)\bs n\bs s \longrightarrow s \cdot s}
			{\infer[\bs_L]{n, n\bs s,n,\ \nabla (n\bs s), \ \nabla (n\bs s)\bs n\bs s \longrightarrow s \cdot s}
				{\infer{n\longrightarrow n}{}
				&
				\infer[\bs_L]{s,n,\ \nabla (n\bs s), \ \nabla (n\bs s)\bs n\bs s \longrightarrow s \cdot s}
					{\infer[\nabla_R]{\ \nabla (n\bs s) \longrightarrow \ \nabla (n\bs s)}
						{\infer[\bs_R]{n\bs s \longrightarrow n\bs s}
							{\infer[\bs_L]{n,n\bs s \longrightarrow  s}
								{\infer{n\longrightarrow n}{}
								&
								\infer{s\longrightarrow s}{}}}}
					&
					\infer[\bs_L]{s,n,n\bs s \longrightarrow s \cdot s}
						{\infer{n\longrightarrow n}{}
						&
						\infer[\cdot_R]{s,s\longrightarrow s\cdot s}
							{\infer{s\longrightarrow s}{}
							&
							\infer{s\longrightarrow s}{}}}}}}}}
%%    	{\infer[perm]{n,\ \nabla (n\bs s),\ \nabla (n\bs s),n,\ \nabla (n\bs s)\bs n\bs s \longrightarrow s \cdot s}
% %  		{\infer[\nabla_L]{n,\ \nabla (n\bs s),n,\ \nabla (n\bs s), \ \nabla (n\bs s)\bs n\bs s \longrightarrow s \cdot s}
%%    			{\infer[\bs_L]{n,n\bs s,n,\ \nabla (n\bs s), \ \nabla (n\bs s)\bs n\bs s \longrightarrow s \cdot s}
%%  				{\infer[]{n\longrightarrow n}{}
%%    				&
%%    				\infer[\bs_L]{s,n,\ \nabla (n\bs s), \ \nabla (n\bs s)\bs n\bs s \longrightarrow s \cdot s}
%    					{\infer[\nabla_R]{\ \nabla (n\bs s) \longrightarrow \ \nabla (n\bs s)}
%    						{n\bs s \longrightarrow n\bs s}
%    					&
%    					\infer[\bs_L]{s, n,n\bs s \longrightarrow s \cdot s}
%    						{\infer[]{n \longrightarrow n}{}
%    						&
%    						\infer[\cdot_R]{s,s \longrightarrow s \cdot s}
%    							{\infer{s\longrightarrow s]{}
%    							&
%    							\infer{s\longrightarrow s]{}}}}}}}}}}
	}\]
	which gives us the following linear map
	\[\begin{array}{c}
		f: N \otimes T_{k_0}(N^*\otimes S) \otimes N \otimes (N\otimes S)^* \otimes N^*\otimes S \longrightarrow S \otimes S 
		\\
		
	\end{array}
	\]
	and again with the following  notational convenience 
	\[\begin{array}{c}
	\overrightarrow{\text{\emph{John}}},\overrightarrow{\text{\emph{Lisa}}}\in N
	\qquad
	\widetilde{\text{\emph{plays guitar}}} \in T_{k_0}(N^* \otimes S)
	\\
	\overrightarrow{\text{\emph{does-too}}} \in (N \otimes S)^*\otimes N^* \otimes S
	\end{array}
	\]
	we can see the action of $f$ simplified to:
	\[
	\begin{array}{l}
	f(\overrightarrow{\text{\emph{John}}} \otimes\widetilde{\text{\emph{plays guitar}}} \otimes \overrightarrow{\text{\emph{Lisa}}} \otimes  \overrightarrow{\text{\emph{does-too}}}) \\ 
	:= \overrightarrow{\text{\emph{plays guitar}}}
	(\overrightarrow{\text{\emph{John}}}) \otimes \overrightarrow{\text{\emph{does-too}}}(\overrightarrow{\text{\emph{plays guitar}}} \otimes \overrightarrow{\text{\emph{Lisa}}}) 
	\end{array}
	\]

	\subsection{Anaphora with Ellipsis}
	Consider the utterance \emph{Kim likes their code. Sam does too}, with $\SLLM$ types:
	\[
	\begin{array}{c}
		\{\text{(\emph{Kim}}:!(\nabla n)),(\text{\emph{Sam}}:!(\nabla n)),
		(\text{\emph{likes}}:!(\nabla(!(\nabla (n\bs s))/n))), 
		(\text{\emph{their}}:!(\nabla(\nabla n\bs n))/n), \\
		(\text{\emph{code}}:n), 
		(\text{\emph{does too}}:\ \nabla (n\bs s)\bs (n\bs s))
		\}.
	\end{array}\]
	This utterance is ambiguous and its ambiguity comes from the fact that one can read it in a \emph{sloppy} way as:  \emph{Kim likes Kim's code. Sam likes \textbf{Sam's} code}, or a \emph{strict} way and as \emph{Kim likes Kim's code. Sam likes \textbf{Kim's} code}.
%	\begin{itemize}
%		\item[] 
%			\textit{Strict}: \emph{Kim likes Kim's code. Sam likes \textbf{Kim's} code}.
%		\item[] 
%			\textit{Sloppy}: \emph{Kim likes Kim's code. Sam likes \textbf{Sam's} code}.
%	\end{itemize}
	In \cite{mcpheat2021} we showed how to distinguish between these two readings syntactically and  semantically in the framework of  \cite{mcpheat2020categorical}. One can do  the same in $\SLLM$  by means of their two distinct derivations, provided in figures \ref{fig:strict} and \ref{fig:sloppy} of the Appendix.

% Experimental  ------------------------ Experimental  ------------------------ Experimental ------------------------
\section{Experimental Evaluation}

To experiment with ellipsis, we use an existing  elliptical sentence similarity dataset\footnote{\url{https://github.com/gijswijnholds/compdisteval-ellipsis/blob/master/datasets/ELLSIM.txt}}, referred to as ELLSIM and originally developed in \cite{wijnholds-sadrzadeh-2019-evaluating};  This dataset extends the transitive sentence similarity dataset\footnote{\url{http://www.cs.ox.ac.uk/activities/compdistmeaning/emnlp2013_turk.txt}}, referred to as KS2013 and developed in  \cite{kartsaklis-sadrzadeh-2013-prior} from transitive sentences of the form \emph{Subject Verb Object} to sentences with verb phrase ellipsis of the form \emph{Subject Verb Object and Subject* does too}. We briefly review these datasets below. 

\subsection{Descriptions of the Datasets}
The KS2013 dataset consists of 108 transitive sentence pairs of the form of \emph{Subject Verb Object}, where the  \emph{Verb Object} pairs come from the dataset ML2010 of  \cite{Mitchell2010}. KS2013 extends these \emph{Verb Object} pairs by adding  subjects to them, according to the following procedure. Each  pair in ML2010 consists of  phrases with  verbs and  objects that are synonymous.  To the first  phrase of each pair,  KS2013  adds  a subject selected from the 5 most frequent subjects of the  verb of the phrase. To the second phrase of the pair, it adds a subject that is synonymous to the --just added-- subject of the first sentence.  As a result of this procedure, KS2013 obtains pairs of sentences that have similar subjects, verbs, and objects to each other. These sentences are rendered as   `highly similar' to each other. As an example, consider the following pair  from the ML2010 dataset:
\[
\langle\mbox{\emph{produce effect, achieve result}}\rangle
\]
This becomes as follows in KS2013:
\[
\langle\mbox{\emph{drug produce effect, medication achieve result}}\rangle
\]
Another example is the following pair  of ML2010:
\[
\langle\mbox{\emph{pose problem, address question}}\rangle
\]
which became as follows in KS2013:
\[\langle\mbox{\emph{study pose problem, paper address question}}\rangle
\]

ML2010 and KS2013 datasets are divided into three bands of pairs: pairs of sentences that have high, low or medium degrees of similarity to each other. The procedure  described above,  returns sentences pairs that are in the high band. The pairs  of the medium band are obtained  by pairing the  sentences of the high band with  sentences  that  only have one  synonymous words with them. Similarly, the pairs in the low band are obtained by pairing the sentences in the high band with sentences that only have one or no synonymous words to them. Examples of pairs of the medium and low bands are provide below:
\begin{eqnarray*}
&\langle&\mbox{\emph{drug produce effect, medication address problem}}\rangle: \mbox{MEDIUM}\\
&\langle&\mbox{\emph{drug produce effect, medication pose result}}\rangle: \mbox{MEDIUM}\\
&\langle&\mbox{\emph{drug produce effect, medication pose problem}}\rangle: \mbox{LOW}\\
&\langle&\mbox{\emph{drug produce effect, employee start work}}\rangle: \mbox{LOW}
\end{eqnarray*}

Similar to ML2010, the KS2013 dataset is  uploaded to AmazonTurk and annotations are obtained for it by multiple participants in a scale from 1 to 7;1 for low similarity, 7 for highly similar. For example, $\langle$\emph{drug produce effect}$\rangle$ and $\langle$\emph{employee start work}$\rangle$ are judged by most annotators to have a low degree of similarity, while $\langle$\emph{drug produce effect}$\rangle$ and $\langle$\emph{medication achieve result}$\rangle$ are considered highly similar. A value in the median of this range represents a medium degree of similarity. 

The ELLSIM dataset is built from KS2013 by choosing two new subjects, let's call denote them   by \emph{Subject*}, for each sentence. \emph{Subject*} is chosen from a list of most frequent subjects of the verb of the sentence, extracted from ukWaC and Wackypedia corpora. The new subjects are appended to a ``does too" ellipsis marker, thus forming the elliptical phrase ``\emph{Subject*} does too". Joining this phrase to the original sentence with a conjunctive word such as ``and"   turns that sentence into the sentence \emph{Subject Verb Object and Subject* does too}, which is a sentence with an elliptical phrase. Here is an example:  consider the KS2013 sentence:  
\[
\mbox{\emph{drug produce effect}}
\]
Two new subjects \textit{plant} and \textit{combination} are chosen for this sentence. These new subjects are used to form the following two elliptical phrases:
\begin{eqnarray*}
&&\mbox{\emph{plant does too}} \\
&&\mbox{\emph{combination does too}}
\end{eqnarray*}
which are then conjoined with the first sentence with an ``and", resulting in the following sentences with elliptical phrases:
\begin{eqnarray*}
&&\mbox{\emph{drug produce effect and plant does too}} \\
&&\mbox{\emph{drug produce effect and combination does too}}
\end{eqnarray*}
\noindent Another example is  the following KS2013 sentence: 
 \[
 \mbox{\emph{medication achieve result}}
 \]
 which is turned into the following two sentences with  elliptical phrases:
 \begin{eqnarray*}
 &&\mbox{\emph{medication achieve result, and patient does too}} \\
 &&\mbox{\emph{medication achieve result and programme does too}}
  \end{eqnarray*}
At the end of this procedure all of the sentences of KS2013 are turned into sentences with elliptical phrases. Each of these newly generated  sentences  are paired with the same sentences as they were before, except  that the sentences they were paired with before are now themselves sentences with elliptical phrases in them. A snapshot of the entries  of the dataset is presented in Table \ref{tab:ellsim}.  The  final dataset contains 432 entries and new annotations are obtained for it following the same procedure as the one that was followed for KS2013, described above. 9800 annotations by 42 different participants are obtained for ELLSIM.

\begin{table}
\centering
\scalebox{0.87}{\begin{tabular}{cc}
\toprule
\textbf{Sentence 1} & \textbf{Sentence 2}
\\ \toprule
\makecell[l]{drug produce effect and combination does too} & medication achieve result and patient does too \\
\makecell[l]{drug produce effect and plant does too} & medication achieve result and patient does too \\
\makecell[l]{drug produce effect and combination does too} & medication achieve result and programme does too \\
\makecell[l]{drug produce effect and plant does too} & medication achieve result and programme does too \\
\bottomrule
\\
\end{tabular}}
\caption{An example of a pair of sentences from ELLSIM dataset. }
\label{tab:ellsim}
\end{table}

\subsection{Building the Representations}

We build vectors for each sentence of ELLSIM and calculate the cosine similarities between pairs of sentences therein. For building these vectors, we first use 300 dimensional pre-trained  \textbf{Word2vec} and \textbf{FastText} word embeddings for the $\ov{Subject}, \ov{Subject^*}$, and $\ov{Object}$ vectors of the dataset. Then, and for the verbs of ELLSIM, we form matrices  $\overline{Verb}$  obtained according to a   formula  developed in \cite{GrefenSadrEMNLP}, called the \emph{Relational} method. This formula is given below:
\[
\overline{Verb} := \sum_i \ov{s}_i \otimes \ov{o_i}
\]
It takes the Kronecker  products of all the subjects $\ov{s}_i$ and objects $\ov{o}_i$  that the verb related to each other in the sentences of the corpus.  The $\ov{s}_i$ and $\ov{o}_i$  vectors of this formulae are built in the same way as the vectors for subjects and objects of the sentences of the dataset, described above.

Leaving the elliptical phrases aside for the moment, in the next step, we describe how we build vector representations for the \emph{Subject Verb Object} parts of the sentences of the dataset. In this step,  the relational verb matrix of the verb of each sentence is \emph{composed} with the    \emph{Subject} and  \emph{Object} vectors of the subject and object of the sentence using the methods developed in   \cite{kartsaklis-etal-2012-unified}, referred to as \emph{Copy Obj} and \emph{Copy Subj}. We  also use the \emph{Frob Add} and \emph{Frob Mult} methods developed in \cite{milajevs-etal-2014-evaluating}.  These compositions are obtained according to the following formulae:

\begin{eqnarray*}
&\mbox{\bf Copy Obj:}&  (\overline{Verb}  \times \ov{Object}) \odot \ov{Subject} \quad \\
&\mbox{\bf Copy Subj:}&  (\overline{Verb}  \times \ov{Subject}) \odot \ov{Object}   \quad \\
&\mbox{\bf Frob Add:}&  \quad  \mbox{\bf Copy Obj} + \mbox{\bf Copy Sub}\\
&\mbox{\bf Frob Mult:}& \quad \mbox{\bf Copy Obj} \odot  \mbox{\bf Copy Sub}\  \\
\end{eqnarray*}

The  ellipsis is taken into account in the final step, where,  we  use the methodology described in the paper to  resolve the ellipsis. Conceptually, this is obtained by  treating the \emph{does too} phrase as a copying map  and applying it to the result of each of the compositional methods described above.  In effect, this procedure will  produce two copies of the verb phrase \emph{Verb Object} of the first part of the sentence and then applies each one of them to the vector of each of the subjects:  \emph{Subject} and \emph{Subject*}. At the end, we obtain a vector representation for the whole sentence with elliptical phrase, that is for  \emph{Subject Verb Object and Subject* does too}. Each compositional method will indeed result in a different vector representation.  For the   \emph{Copy Obj}  method we have the following representation:

\small\[
\mbox{\bf Multiplex}:  ((\overline{Verb} \times \ov{Object}) \odot \ov{Subject}) +  ((\overline{Verb} \times \ov{Object}) \odot \ov{Subject^*})
\]
The formula for the \emph{Copy Subj}  method is obtained from the above by exchanging  $\ov{Object}$ and $\ov{Subject}$; \emph{Frob Add} is the sum of \emph{Copy Obj} and \emph{Copy Subj}, \emph{Frob Mult} is their product.

\begin{table}[t!]
  \centering
  \scalebox{0.9}{\begin{tabular}{lcc} 
  \toprule
    \textbf{Method}&\textbf{Embeddings}& \quad  \textbf{Results}  \\
    \midrule
        \textbf{Copy Subject} &  \makecell[r]{word2vec\\fasttext} &  \makecell[r]{0.644\\0.591} \\
    \midrule
    \textbf{Copy Object} &  \makecell[r]{word2vec\\fasttext} &  \makecell[r]{0.604\\0.599} \\
    \midrule
    \textbf{Frobenius Add.} &  \makecell[r]{word2vec\\fasttext} &  \makecell[r]{0.653\\0.610} \\
    \midrule
    \textbf{Frobenius Mult.} &  \makecell[r]{word2vec\\fasttext} &  \makecell[r]{0.587\\0.579} \\
   \hline
   \hline
   Baselines\\
  \hline
  \hline
  \textbf{Verb Only Vector} & \makecell[r]{word2vec\\fasttext} &  \makecell[r]{0.583\\0.651}\\
 \midrule
 \textbf{Verb Only Tensor} & \makecell[r]{word2vec\\fasttext} &  \makecell[r]{0.566\\0.533}\\
 %\midrule
  %\textbf{Additive All} & \makecell[r]{word2vec\\fasttext} &  \makecell[r]{0.717\\0.738}\\
 \midrule
 \textbf{Additive} & \makecell[r]{word2vec\\fasttext} &  \makecell[r]{0.768\\0.783}\\
 \midrule
 \textbf{BERT phrase}  & & 0.575 \\
 \bottomrule
 \\
 \end{tabular}}
\caption{Experiment 1: Spearman $\rho{}$ scores}
 \label{tab:exp1}
\end{table}

%\begin{table}[t!]
%  \centering
%  \scalebox{0.9}{\begin{tabular}{lccccc} 
%  \toprule
%    &&&basis & basis &\\
%    \textbf{Method}&{}& full &  copy(a) &  copy(b) &  $\mathbf{k}$-extension \\
%    \midrule
%        \textbf{Copy Subject} &  \makecell[r]{word2vec\\fasttext} &  \makecell[r]{0.644\\0.591} &  \makecell[r]{0.565\\0.570} & \makecell[r]{0.626\\0.595} &  \makecell[r]{0.644\\0.591}\\
%    \midrule
%    \textbf{Copy Object} &  \makecell[r]{word2vec\\fasttext} &  \makecell[r]{0.604\\0.599} &  \makecell[r]{0.524\\0.503} & \makecell[r]{0.488\\0.432} &  \makecell[r]{0.604\\0.599}\\
%    \midrule
%    \textbf{Frobenius Add.} &  \makecell[r]{word2vec\\fasttext} &  \makecell[r]{0.653\\0.610} &  \makecell[r]{0.579\\0.548} & \makecell[r]{0.583\\0.509} &  \makecell[r]{0.653\\0.610}\\
%    \midrule
%    \textbf{Frobenius Mult.} &  \makecell[r]{word2vec\\fasttext} &  \makecell[r]{0.587\\0.579} &  \makecell[r]{0.555\\0.516} & \makecell[r]{0.543\\0.492} &  \makecell[r]{0.587\\0.579}\\
%   \hline
%   \hline
%   Baselines\\
%  \hline
%  \hline
%  \textbf{Verb Only Vector} & \makecell[r]{word2vec\\fasttext} &  \makecell[r]{0.583\\0.651}\\
% \midrule
% \textbf{Verb Only Tensor} & \makecell[r]{word2vec\\fasttext} &  \makecell[r]{0.566\\0.533}\\
% %\midrule
%  %\textbf{Additive All} & \makecell[r]{word2vec\\fasttext} &  \makecell[r]{0.717\\0.738}\\
% \midrule
% \textbf{Additive} & \makecell[r]{word2vec\\fasttext} &  \makecell[r]{0.768\\0.783}\\
% \midrule
% \textbf{BERT phrase}  & 0.575 &&&\\
% \bottomrule
% \\
% \end{tabular}}
% \label{tab:exp1}
%\caption{Experiment 1: Spearman $\rho{}$ scores}
%\end{table}

After a representation is built for each sentence of the dataset, we measure the degree of  similarity  between sentences of each pair. Traditionally, this is done by computing the cosine distance between the sentence vectors. The distances are then put  against  the degrees of similarity obtained from human annotations for each pair.  In order to form a judgement about whether the representations were good or bad, we perform three statistical tests, described in the following three subsections.  

For each of these tests, we compare  different  vector embeddings and compositional methods with each other  and  with a  few  baselines. We describe our baselines below. Firstly, we consider a non-compositional baseline that takes the representation of each sentence to be the same as the representation of its verb, thus ignoring all the other constituents of the sentence. For this, we use the  \emph{Verb} only vector  and \emph{Verb} only tensor representations of the \emph{Verb}. For  the \emph{Verb} only tensor, we use the \emph{Relational} matrix of the verb as described above. Another of our baselines is the  Additive model; this model is compositional but does not take the grammatical structure into account. It is obtained by adding the vectors of  \emph{Subject},  \emph{Verb}, \emph{Object}, and \emph{Subject*} of the sentence. Our next baseline is the BERT base model; for this we used pre-trained embeddings of dimension 768 extracted from the second-to-last hidden layer of all tokens in the sentence with average pooling.

%Finally, we also compare with the copying methodology that we had developed in our previous paper \cite{mcpheat2020categorical}. This method was called  $\mathbf{k}$-extension and was a linear approximation of the   {\bf full} copying map developed here. Following  previous work,    we instantiate  $\mathbf{k}$ to 1, this results in the following formula, for the \emph{Copy Obj} method: 
%
%\footnotesize\begin{eqnarray*}
%\mathbf{k}\mbox{-extension}:&  ((\overline{V} \times \ov{Obj}) \odot \ov{Sub1} +  (\mathbf{1} \odot \ov{Sub2})) +  ((\mathbf{1} \odot  \ov{Sub1})  + (\overline{V} \times \ov{Obj}) \odot\ov{Sub2})\\
%= & ((\overline{V} \times \ov{Obj}) \odot \ov{Sub1} +  \ov{Sub2}) +  (\ov{Sub1}  + (\overline{V} \times \ov{Obj}) \odot\ov{Sub2})
%\end{eqnarray*}\normalsize
%
%\noindent In the above,  the first summand  on its own is also a method of previous work and what we referred to as basis-copy(a) in \cite{mcpheat2020categorical}. Similarly,  the second summand is a method of its own and what we referred to as basis-copy(b)  in \cite{mcpheat2020categorical}.  As for the {\bf full} copying map, the  formula for the \emph{Copy Subj} is obtained from the above by exchanging  $\ov{Object}$ with $\ov{Subject}$; \emph{Frob Add} is the sum of \emph{Copy Obj} and \emph{Copy Subj}, \emph{Frob Mult} is their product. 

 As an example, consider the pair of sentences:

\begin{itemize}
	\item[]
		 S1: \emph{drug produce effect and combination does too}
	\item[] 
		 S2: \emph{medication achieve result and patient does too}
\end{itemize}

\noindent Following our  methods, we obtain the following vectors for them:\\

\begin{itemize}
	\item 
	
	\noindent Our copying  map applied to the  \emph{Copy Obj}:
	\\\\
	\scriptsize
	$\ov{S1}$ =  $((\overline{produce} \times \ov{effect}) \odot \ov{drug}) +  ((\overline{produce} \times \ov{effect}) \odot \ov{combination})$
	
	$\ov{S2}$ =  $((\overline{achieve} \times \ov{result}) \odot \ov{medication}) +  ((\overline{achieve} \times \ov{result}) \odot \ov{patient})$
	\\
	\normalsize

	\noindent Our copying  map applied to the  \emph{Copy Subj}:
	\\\\
	\scriptsize
	$\ov{S1}$ =  $((\overline{produce} \times \ov{drug}) \odot \ov{effect}) +  ((\overline{produce} \times \ov{combination}) \odot \ov{effect})$
	
	$\ov{S2}$ =  $((\overline{achieve} \times \ov{medication}) \odot \ov{result}) +  ((\overline{achieve} \times \ov{patient}) \odot \ov{result})$
	\\
	\normalsize

%	
%	\item
%	
%	\noindent  The ($\mathbf{k}$-extension with $\mathbf{k}$=1) map applied to \emph{Copy Obj}:
%	\\
%	\\
%	\scriptsize
%	$\ov{S1}$ =  $((\overline{produce} \times \ov{effect}) \odot \ov{drug} +  \ov{combination}) +  (\ov{drug}  + (\overline{produce} \times \ov{effect}) \odot\ov{combination})$
%	
%	$\ov{S2}$ =  {$((\overline{achieve} \times \ov{result}) \odot \ov{medication} +  \ov{patient}) +  (\ov{medication}  + (\overline{achieve} \times \ov{result}) \odot\ov{patient})$}\\
%	\normalsize
%	
	\item	
	 
	\noindent The Additive baseline:
	\\
	\\
	\scriptsize
	$\ov{S1}$ =  {$\ov{drug} + \ov{produce} + \ov{effect} + \ov{combination}$}
	
	$\ov{S2}$ =  {$\ov{medication} + \ov{achieve} + \ov{result} + \ov{patient}$}
	\normalsize
\end{itemize}

\subsection{Spearman Experiment}

Following \cite{wijnholds-sadrzadeh-2019-evaluating}, we calculate Spearman’s rank correlation coefficient between the cosine similarity scores of pairs of sentences of ELLSIM and the average human annotation judgments. This is a value between -1 and +1. Results are presented in Table \ref{tab:exp1}. Our first observation is that the best performing model was the Additive baseline. It achieved the highest Spearman correlation score of 0.783 in the  \textbf{FastText}  space. Between all our compositional methods, the highest correlation score was 0.653 with the Frobenius Additive and in the \textbf{Word2vec} space. All of the compositional models outperformed the BERT phrasal model but not the Additive baseline model. In both cases, the results are impressive especially since BERT is considered to be the state-of-the-art context-based neural model, trained on a huge dataset, and consisting of millions of parameters. BERT created a breakthrough in the field of NLP by providing greater results in many NLP tasks, such as question answering, text generation, sentence classification, and many more besides. Getting better results than BERT in this sentence similarity task is considered a good achievement, but of course, we fell short of a very simple additive model.

%It's important to note that BERT stands for Bidirectional Encoder Representation from Transformer. It is the state-of- the-art neural embedding model published by Google, trained on a huge dataset, and consists of millions of parameters. It has created a breakthrough in the field of NLP by providing greater results in many NLP tasks, such as question answering, text generation, sentence classification, and many more besides. One of the major reasons for the success of BERT is that it is a context-based embedding model, unlike other popular embedding models, such as Word2Vec, which are context-free. In brief, BERT provides different meanings for words in different context.

\begin{table}[t!]
  \centering
  \scalebox{0.9}{\begin{tabular}{lcc} 
  \toprule
    \textbf{Method}& \textbf{Embeddings} & \quad \textbf{Results} \\
    \midrule
    \textbf{Copy Subject} &  \makecell[r]{word2vec\\fasttext} &  \makecell[r]{15.99\\15.09}\\
    \midrule
    \textbf{Copy Object} &  \makecell[r]{word2vec\\fasttext} &  \makecell[r]{14.82\\14.40} \\
    \midrule
    \textbf{Frobenius Add.} &  \makecell[r]{word2vec\\fasttext} &  \makecell[r]{14.86\\14.24} \\
    \midrule
    \textbf{Frobenius Mult.} &  \makecell[r]{word2vec\\fasttext} &  \makecell[r]{15.97\\14.85} \\
   \hline
   \hline
   Baselines\\
  \hline
  \hline
 \textbf{Verb Only Vector} & \makecell[r]{word2vec\\fasttext} &  \makecell[r]{15.05\\15.13}\\
 \midrule
 \textbf{Verb Only Tensor} & \makecell[r]{word2vec\\fasttext} &  \makecell[r]{15.01\\14.56}\\
 \midrule
  \textbf{Additive} & \makecell[r]{word2vec\\fasttext} &  \makecell[r]{14.32\\14.32}\\
 \midrule
 \textbf{BERT phrase}  & & 13.76 \\
 \bottomrule
 \\
\end{tabular}}
\caption{Experiment 2: Student's $t$-test results }
\label{tab:exp2}
\end{table}

%\begin{table}[t!]
%  \centering
%  \scalebox{0.9}{\begin{tabular}{lccccc} 
%  \toprule
%    &&&basis & basis &\\
%    \textbf{Method}&{}& full &  copy(a) &  copy(b) &  $\mathbf{k}$-extension \\
%    \midrule
%    \textbf{Copy Subject} &  \makecell[r]{word2vec\\fasttext} &  \makecell[r]{15.99\\15.09} &  \makecell[r]{15.88\\15.18} & \makecell[r]{16.13\\15.17} &  \makecell[r]{15.99\\15.09}\\
%    \midrule
%    \textbf{Copy Object} &  \makecell[r]{word2vec\\fasttext} &  \makecell[r]{14.82\\14.40} &  \makecell[r]{15.67\\15.01} & \makecell[r]{15.78\\15.22} &  \makecell[r]{14.82\\14.40}\\
%    \midrule
%    \textbf{Frobenius Add.} &  \makecell[r]{word2vec\\fasttext} &  \makecell[r]{14.86\\14.24} &  \makecell[r]{15.35\\14.55} & \makecell[r]{15.54\\14.81} &  \makecell[r]{14.86\\14.24}\\
%    \midrule
%    \textbf{Frobenius Mult.} &  \makecell[r]{word2vec\\fasttext} &  \makecell[r]{15.97\\14.85} &  \makecell[r]{16.52\\15.47} & \makecell[r]{16.29\\15.10} &  \makecell[r]{15.97\\14.85}\\
%   \hline
%   \hline
%   Baselines\\
%  \hline
%  \hline
% \textbf{Verb Only Vector} & \makecell[r]{word2vec\\fasttext} &  \makecell[r]{15.05\\15.13}\\
% \midrule
% \textbf{Verb Only Tensor} & \makecell[r]{word2vec\\fasttext} &  \makecell[r]{15.01\\14.56}\\
% \midrule
%  \textbf{Additive} & \makecell[r]{word2vec\\fasttext} &  \makecell[r]{14.32\\14.32}\\
% \midrule
% \textbf{BERT phrase}  & 13.76 &&&\\
% \bottomrule
% \\
%\end{tabular}}
%\caption{Experiment 2: Student's $t$-test results }
%\label{tab:exp2}
%\end{table}

\subsection{Student's $t$-test Experiment}
Following \cite{polajnar2014b} we  examine the difference between two group means: the average human annotation judgments and the cosine similarity scores via student's $t$-test. First, we match each sentence pair in the dataset with its corresponding group of sentences. Then, we calculate the $t$-score for each group and compute the overall average. The smaller the $t$-score, the more similarity exists between the two sets. Results are presented in Table \ref{tab:exp2}. Here, BERT did the best with the lowest t-score, which was 13.76. Then, we had our Frobenius Add. model with a t-score of 14.24. In the third place, came  the Additive baseline with a t-score of 14.32.  Our best results were achieved  best while using \textbf{FastText} embeddings.

\begin{table}[t!]
  \centering
  \scalebox{0.9}{\begin{tabular}{lcc} 
  \toprule
    \textbf{Method}&\textbf{Embeddings} &\quad  \textbf{Results} \\
    \midrule
        \textbf{Copy Subject} &  \makecell[r]{word2vec\\fasttext} &  \makecell[r]{71.18\%\\73.44\%} \\
    \midrule
    \textbf{Copy Object} &  \makecell[r]{word2vec\\fasttext} &  \makecell[r]{76.33\%\\74.19\%} \\
    \midrule
    \textbf{Frobenius Add.} &  \makecell[r]{word2vec\\fasttext} &  \makecell[r]{76.04\%\\73.67\%} \\
    \midrule
    \textbf{Frobenius Mult.} &  \makecell[r]{word2vec\\fasttext} &  \makecell[r]{69.79\%\\71.24\%} \\
   \hline
   \hline
   Baselines\\
  \hline
  \hline
  \textbf{Verb Only Vector} & \makecell[r]{word2vec\\fasttext} &  \makecell[r]{76.44\%\\78.53\%}\\
 \midrule
 \textbf{Verb Only Tensor} & \makecell[r]{word2vec\\fasttext} &  \makecell[r]{73.30\%\\73.30\%}\\ 
 \midrule
 \textbf{Additive} & \makecell[r]{word2vec\\fasttext} &  \makecell[r]{82.00\%\\81.77\%}\\
 \midrule
 \textbf{BERT phrase}  & & 75.93\% \\
 \bottomrule
\\
\end{tabular}}
\caption{Experiment 3: Classification Accuracy results }
\label{tab:experiment3}
\end{table}

%\begin{table}[t!]
%  \centering
%  \scalebox{0.9}{\begin{tabular}{lccccc} 
%  \toprule
%    &&&basis & basis &\\
%    \textbf{Method}&{}& full &  copy(a) &  copy(b) &  $\mathbf{k}$-extension \\
%    \midrule
%        \textbf{Copy Subject} &  \makecell[r]{word2vec\\fasttext} &  \makecell[r]{71.18\%\\73.44\%} &  \makecell[r]{69.27\%\\72.51\%} & \makecell[r]{69.97\%\\73.61\%} &  \makecell[r]{71.18\%\\73.44\%}\\
%    \midrule
%    \textbf{Copy Object} &  \makecell[r]{word2vec\\fasttext} &  \makecell[r]{76.33\%\\74.19\%} &  \makecell[r]{73.78\%\\71.12\%} & \makecell[r]{67.48\%\\64.35\%} &  \makecell[r]{76.33\%\\74.19\%}\\
%    \midrule
%    \textbf{Frobenius Add.} &  \makecell[r]{word2vec\\fasttext} &  \makecell[r]{76.04\%\\73.67\%} &  \makecell[r]{75.87\%\\72.16\%} & \makecell[r]{70.49\%\\69.85\%} &  \makecell[r]{76.04\%\\73.67\%}\\
%    \midrule
%    \textbf{Frobenius Mult.} &  \makecell[r]{word2vec\\fasttext} &  \makecell[r]{69.79\%\\71.24\%} &  \makecell[r]{69.97\%\\73.09\%} & \makecell[r]{67.13\%\\65.80\%} &  \makecell[r]{69.79\%\\71.24\%}\\
%   \hline
%   \hline
%   Baselines\\
%  \hline
%  \hline
%  \textbf{Verb Only Vector} & \makecell[r]{word2vec\\fasttext} &  \makecell[r]{76.44\%\\78.53\%}\\
% \midrule
% \textbf{Verb Only Tensor} & \makecell[r]{word2vec\\fasttext} &  \makecell[r]{73.30\%\\73.30\%}\\ 
% \midrule
% \textbf{Additive} & \makecell[r]{word2vec\\fasttext} &  \makecell[r]{82.00\%\\81.77\%}\\
% \midrule
% \textbf{BERT phrase}  & 75.93\% &&&\\
% \bottomrule
%\\
%\end{tabular}}
%\caption{Experiment 3: Classification Accuracy results }
%\label{tab:experiment3}
%\end{table}

\subsection{Classification Experiment}
Following \cite{mcpheat2020categorical,mcpheat2021}, we also perform a classification experiment. We compare the average human annotation judgments and the cosine similarity scores between triplets of sentences. An example of these triplets is shown below. A new dataset\footnote{\url{https://cutt.ly/QmacrwT}} was generated from the original dataset and used in this experiment. %These triplets were generated from the original dataset for the new classification experiment of this paper and can be accessed from here ???. 

\begin{itemize}
	\item[] 
		Sentence 1: \emph{drug produce effect and combination does too}
	\item[] 
		Sentence 2: \emph{employee start work and team does too}
	\item[] 
		Sentence 3: \emph{committee consider matter and panel does too}
\end{itemize}

A correct classification is obtained when both the average human annotation score and cosine similarity  between one pair of the triplet is higher than the other pair. In other words, the average human annotation score and cosine similarity agree. % and the cosine similarity between the same first pair is also higher than that of the same second pair. %Similarly, if the average human annotation score between Pair 2 is higher than that of Pair 1 and the cosine similarity between Pair 2 is higher than that of Pair 1, we count this as a correct classification. 
An overall accuracy is computed by counting the number of correct classifications out of these comparisons. For instance, the cosine similarity measure between Sentences 1 and 2 above using the Additive model and \textbf{Word2vec} embeddings is 0.104, while that between Sentences 1 and 3 above is 0.195. On the other hand, the average human annotation score between Sentences 1 and 2 is  2.217 while that between Sentences 1 and 3 is 2.695. As a result, the cosine similarity score between Sentences 1 and 3 is higher than that between 1 and 2 and similarly the human annotation scores. This match is counted as a correct classification.

Table \ref{tab:experiment3} shows that similar to the Spearman's  test,  the Additive model  achieved the highest accuracy of 82.00\%, even higher than BERT, whose accuracy was 75.93\%.  Our best model was  \textit{Copy Obj}  with an accuracy of 76.33\%. This was in fact the second best accuracy after the Additive model with values  even higher than BERT's.
%Note that for the Verb only baselines, we had to modify the dataset, as some triplets had the same verb but different subjects and objects. The resulting dataset obtained by dropping these cases only contained 191 triplets; the original  dataset had 1728 triplets.
%Here, we examine the similarity between pairs instead of groups. We built a new dataset of the form of triplets where each triplet consists of two pairs: Pair 1 (\textit{Sentence 1, Sentence 2}) and Pair 2 (\textit{Sentence 1, Sentence 3}). For each group such as in Table 2, a set of 12 unique triplets is generated. 

\section{Conclusion and Outlook}
We develop a categorical and a vector  space semantics for a decidable version of an extension of Lambek Calculus used to model parasitic gaps and coreference.  This enables us to rely on a decidable logic to produce the type-logical derivations of these phenomena and develop compositional type-driven vector representations for them. We also achieved what  previously was not possible: working with full copies of types, rather than their non-linear approximations.  We perform three statistical  tests using data from large scale corpora and a distributional verb phrase elliptical similarity task. In two out of the three experiments,  our model  outperformed the state-of-the-art BERT. We think this is impressive since  BERT has perform very well on downstream NLP tasks, which require complex modelling of structural relationships \cite{lin2019}. Our models, however, were outperformed by a very simple additive model.  The difference between our best model and the Additive model was were however,  very small  and in the range 0.1 or around 10\%.  

What can we learn from these results? That we beat BERT in one test and came very close to it in another,  is quite impressive;  that in two tests, we were beaten by a very simple additive model that ignores all structure, is disappointing, however, we did come very close to this model. We can for sure say that it is worth taking the grammatical and discourse structures into account when building vector representations for sentences. The resulting representations might not perform better than other simple or state of the art models, but will nonetheless come very close to it. So if our aim is to represent structures of sentences as well as their distributional properties, our model is the one to go for. One might not obtain the best results, but the loss is very little. In our experiments, the differences between our best model and the BERT or Additive models were only about 0.1 in $t$-test  and Spearman and 10\% in accuracy. A qualitative analysis only made sense for our classification test, the other tests were made over the patterns of the whole dataset. Going through the instances of the classification test, we observed that in pairs of sentences such as $\langle$ \emph{medication achieve result and patient does too, drug produce effect and combination does too}$\rangle$  where the disambiguation is more  clear, the additive model does better. In this example, the sentence pair belonged to the HIGH similarity band,  the cosine of the angle between its sentences in the additive model was 0.74, whereas the cosine of the angle in our model was 0.72. The human annotations where also on the high side with an average score of 4.26 (recall that these are between 1 and 7). In examples where the disambiguation was less clear, e.g.  in the pair $\langle$\emph{user send message and server does too, man hear word and listener does too}$\rangle$, from the MEDIUM similarity band, our model worked better by producing a cosine of 0.33 as opposed to the additive model which resulted in a cosine of  0.13. The average human score in this case was 3.33. 
	
%	\begin{color}{red} A qualitative analysis is also lacking: on what sentences does the approach perform best (or worst), in particular when compared with non-compositional approaches?\end{color}
	
	A few directions to pursue for future work are as follows. Firstly,  the theoretical challenge  of finding complete  vector space models   remains an open challenge. Investigating completeness possibly through centres of monoidal functors over finite dimensional vector spaces with a Hopf structure  is our work in progress. Then, experimenting with the setting on the Winograd Schema Challenge \cite{levesque2012} using  the plausibility models of \cite{polajnar2014b}  is  work in progress. Still on the experimental side, we have  compared the current results to its   approximations, as developed in previous work \cite{mcpheat2020categorical}. The  basis-copying operations  used before, the so called Frobenius copying maps, did not perform well. Their additive combination using a method we called $k$-extension, however, produced comparable results. Here  we observed an anomaly that we could not resolve in this paper and which we leave to further experimentation. 
	
	%First, although $\SLLM$ allows us to separate copying from moving, in our natural language derivations they always appear together. Further, . We 
	%This is due to the method we handle our examples. coreference type-logically. The conventional treatment of coreference is to copy the referent word, and move one of the copies to the referencing word and then identify them \cite{wijnholds-sadrzadeh-2019-evaluating}. Simply because one will always copy, then move for coreference, and even for the parasitic gap example below, we will always need to preceed the $!$ with $\nabla$. 
	%
	%One may wish to consider a variation of $\SLLM$ where we combine the two rules, and investigate decidability and completeness.
	
	%e.g. replace the $!_L$ rule with
%		\[\infer[!_L']
%			{\Gamma_1, !A, \Gamma_2 \longrightarrow B}
%			{\Gamma_1, \nabla A,\nabla A,\ldots, \nabla A, \Gamma_2 \longrightarrow B}
%		\]
		%which we initially pursued, but turned out to be a lot of work for a structure that is easily achievable in $\SLLM$ by nesting $\nabla$ inside of $!$. 
		%Second, obtaining a complete model rests on the fact that ... . 

\bibliographystyle{plain}
\bibliography{references}

\section{Appendix}
	\begin{figure}[h!]
	\includegraphics[width=\columnwidth]{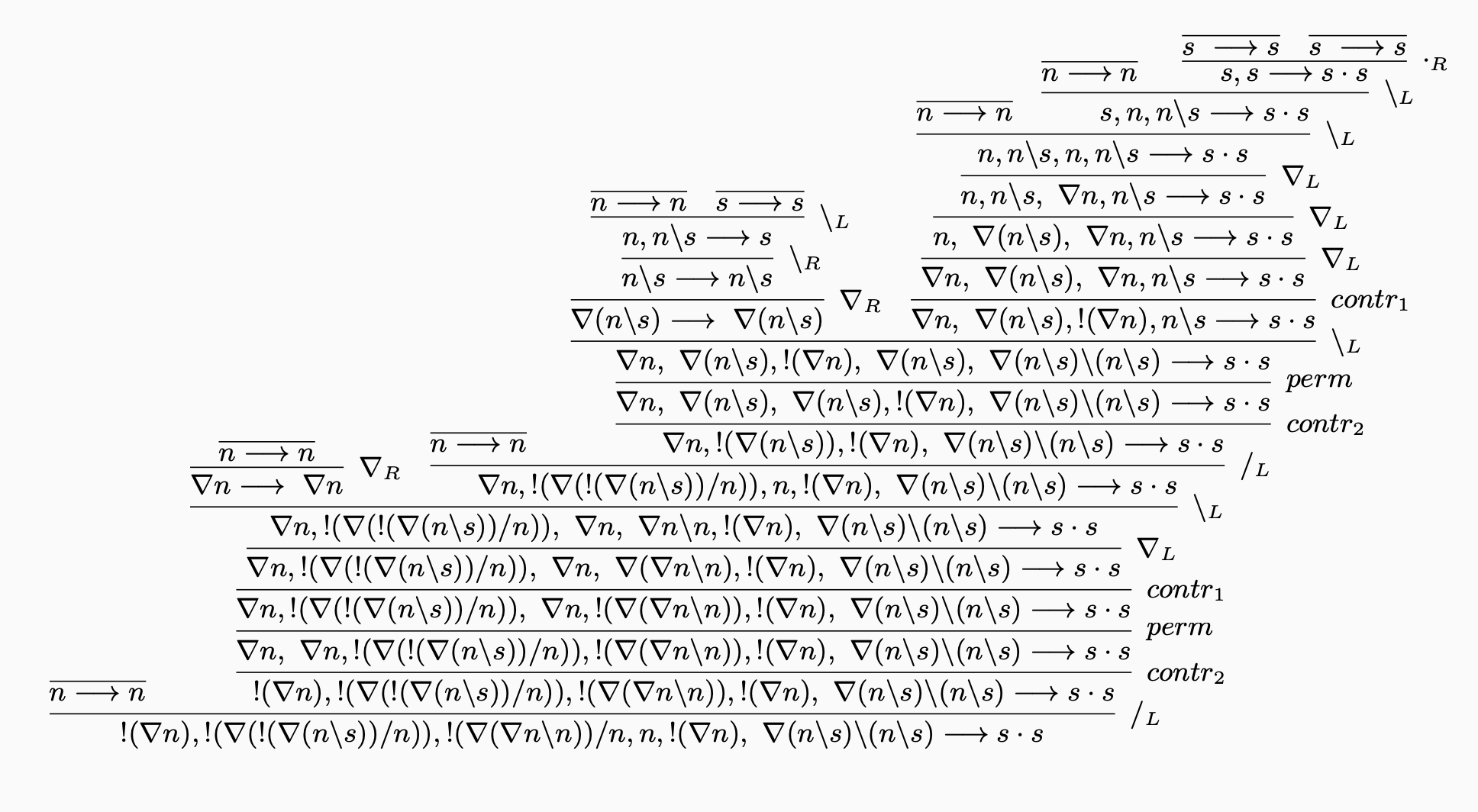}
	\caption{The strict reading of \emph{Kim likes their code. Sam does too}.}\label{fig:strict}
	\end{figure}
	
	\begin{figure}[b!]
	\centering
	\includegraphics[width=\columnwidth]{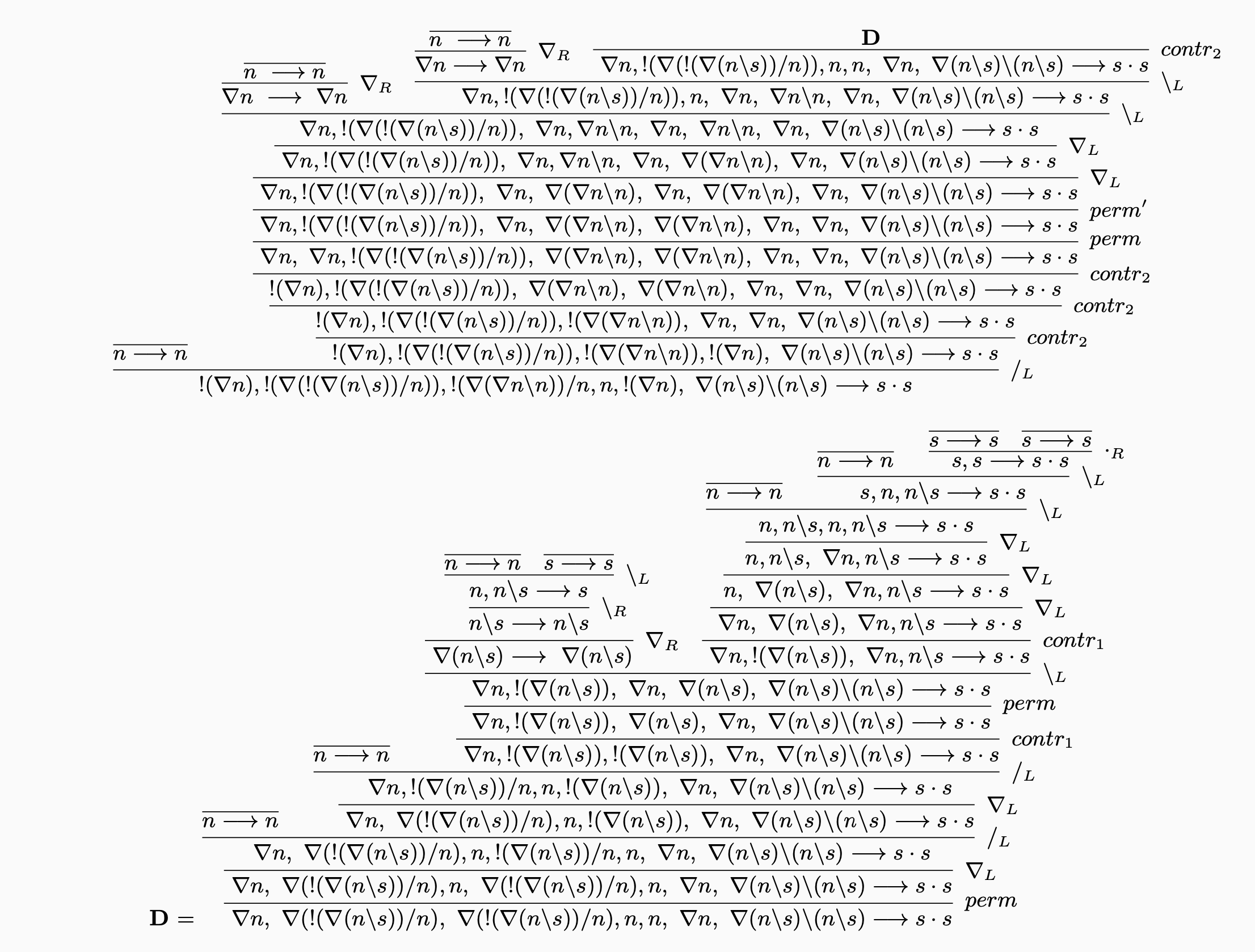}
	\caption{The sloppy reading of \emph{Kim likes their code. Sam does too}.}\label{fig:sloppy}
	\end{figure}
\end{document}